\documentclass[12pt,journal,draftcls,a4paper,onecolumn]{IEEEtran}
\usepackage{cite}
\usepackage{graphicx}
\usepackage{amsmath}
\usepackage{array}
\usepackage{tabularx}
\usepackage{dcolumn}
\usepackage{amssymb}
\usepackage{multicol}
\usepackage{subfigure}
\usepackage{cite}
\usepackage{amsmath}
\usepackage{amsfonts}
\usepackage{subfigure}
\usepackage{lipsum}
\usepackage{ntheorem}

\newtheorem{defi}{Definition}

\begin{document}
\title{
An Overflow Problem in  Network Coding for   Secure Cloud Storage
}
\author{
Yu-Jia Chen,~\IEEEmembership{Member,~IEEE} and Li-Chun Wang,~\IEEEmembership{Fellow,~IEEE,}
\\National Chiao Tung University, Taiwan
\\Email: allan920693@g2.nctu.edu.tw and lichun@cc.nctu.edu.tw}

\maketitle

\begin{abstract}

In this paper we define the overflow problem  of a network coding storage system in which the encoding parameter and the storage parameter are mismatched. Through analyses and experiments, we first show  the impacts of the overflow problem in a network coding scheme, which not only waste storage spaces, but also degrade coding efficiency. To avoid the overflow problem, we then develop the network coding based secure storage (NCSS) scheme.
Thanks to considering both security and storage requirements in encoding procedures and
distributed architectures,  the NCSS can improve the performance of a cloud storage system from both the aspects of storage cost and coding processing time.
We analyze the maximum allowable stored encoded data under the perfect secrecy criterion, 
and provide the design guidelines for the secure cloud storage system to enhance coding efficiency  and achieve the minimal storage cost.

%

\end{abstract}


\section{Introduction}  \label{Introduction}


Network coding is an attractive solution for secure cloud storage because of  achieving the unconditional security.
As long as protecting  partial network coded data,  a non-collusive eavesdropper cannot decode any symbol even with huge computing power  for infinite time \cite{Oliveira2010}.
In principle, network coding simply mixes the data from different network nodes  based on the well-designed linear combination rules.
Hence, almost incurring no bandwidth expansion is another advantage of network coding.

\begin{figure}[t]
\centering
\includegraphics[width=0.9\textwidth]{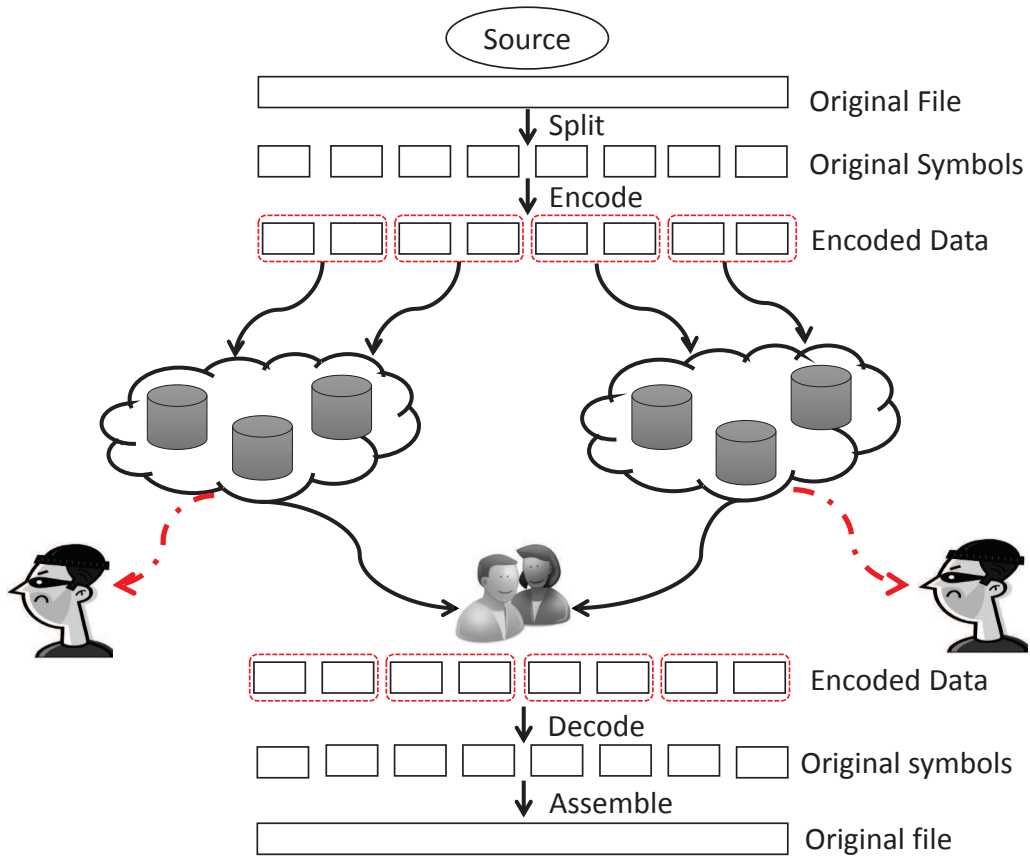}\\
\caption{An example of secure cloud storage using network coding.} \label{fig2}
\end{figure}

A secure cloud storage using network coding is illustrated in Fig. \ref{fig2}.
The original file is split into smaller chunks of symbols, each of which is encoded by Vandermonde matrix \cite{Klinger_21}. The different subsets of the encoded data are stored to two cloud databases. A legitimate user with access to two cloud databases can recover the entire original file.
However, an eavesdropper with access to only one of the two cloud databases is unable to decode any of the original symbols \cite{oliveira2012coding}. In summary, the network coding storage system consists of three procedures: splitting, encoding, and distributing to storage nodes.

Nevertheless, a secure network coding storage system may encounter a practical design issue when encoding parameters (such as the size of encoding matrix) is not jointly designed with the storage parameters (such as the storage size per node). If the mismatch between encoding and storage parameters occurs, it  can cause bandwidth expansion and redundant computation cost. We coin the term  the overflow problem for secure network coding storage system in this paper because the mismatch of encoding and storage parameters results in extended length
of coded data in the format of digits.


Table \ref{binary_example} shows an example of the overflow problem for binary digits, where ${\bf{A}}$ is the encoding matrix for network coding, ${b_i}$ and ${c_i}$ are the original data, and network coded data, respectively.
Assume that  ${c_1}$ and ${c_3}$ are stored in the first database, and $c_2$ is stored in the second database.
We can see that the bit length of coded ${c_i}$  is larger than
that of ${b_i}$ for $i = 1 \sim 3$. Also, the minimum bit length required for storing  ${c_1}$ and ${c_3}$ is three  in the first database, but in the second database the minimum bit length for storing $c_2$  is two.

\begin{table*}
\vspace*{0cm} \caption{\bfseries Example of the overflow problem for binary digits} \label{binary_example}
\begin{center}
\vspace*{-0.1cm}\scriptsize
\renewcommand{\arraystretch}{2}
\begin{tabular}{|l|l|l|l|l|l|l|l|} \cline{1-3}

 \multicolumn{1}{|c}{${\bf{A}}$}
& \multicolumn{1}{|c}{${\bf{b}}$}
&\multicolumn{1}{|c|}{${\bf{c}}$} \\
\hline
%
%

\multicolumn {1}{|c}  {$\left[ {\begin{array}{*{20}{c}}
1&1&1\\
5&1&2\\
6&1&4
\end{array}} \right]$}
&\multicolumn{1}{|c}  {$(0,1,1)^T$}
&\multicolumn{1}{|c|} {$(2,3,5)^T$=$(10,11,101)^T$} \\
\hline

\end{tabular}
\end{center}
\end{table*}

To overcome the overflow problem, we  propose a systematic  design methodology to calculate the important system parameters of a network-coded cloud storage system.
The major objective of the proposed scheme is to provide correct mapping between the encoding parameters and the storage parameters.
The contributions of this work are explained as follows.
\begin{itemize}

\item
Formulate the overflow problem  of a network-coded cloud storage system, and perform experiments to show the impacts of the overflow problem, which not only waste storage spaces, but degrade computational efficiency.

\item
Propose a network coding based secure storage (NCSS) scheme to  solve the overflow problem. To our best knowledge, the overflow problem for a network-coded storage system has not been investigated in the literature yet.

\item
Derive the upper bound of the amount of encoded data that can be stored in cloud databases to achieve the unconditional security level (i.e., perfect secrecy).
Based on the derived upper bound, we present the analysis of storage cost minimization in the proposed NCSS scheme subject to different security levels.

\item
Finally, based on the experimental results, we suggest the design guidelines for determining important system parameters (e.g., the size of the encoding matrix) to accelerate the network coding process.
\end{itemize}

The rest of this paper is organized as follows. Section II describes related works. In Section III, we formulate the overflow problem in cloud storage using network coding. In Section IV, we present the NCSS scheme. In Section V we give the security and storage analyses of the proposed scheme. Section VI shows the experimental results. Finally, we give our concluding remarks in Section VII.

\section{Background and Related Work}  \label{Background}

Network coding can be viewed as a generalized store-and-forward network routing principle. Messages from different  source nodes are combined and regenerated at the intermediate nodes according to algebraic encoding. Besides the well-known advantages of throughput enhancement \cite{6134039,zeng2014throughput,wu2006distributed} and data robustness \cite{Fragouli2006}, the recent studies on network coding focus on reliability and security enhancement.

\subsection{Network Coding for Data Recovery}
Network coding can improve the efficiency of data recovery process when storage nodes fail in distributed storage systems.
It was proved that the data recovery problem of distributed storage systems is equivalent to the multicasting problem of network routing \cite{Dimakis}.
The authors of \cite{Hu2010} designed a cooperative network coding recovery mechanism for multiple node failures.
A proxy-based  multiple cloud storage system with the feature of fault tolerance  was built based on the network coding storage scheme \cite{Hu2012}.
A network coding method called \textsl{Regenerating Code} was proposed to improve the repair process of distributed storage systems \cite{papailiopoulos2012simple}.
Different from erasure coding, the repaired data fragments are mixed in intermediate nodes, thereby reducing the repair bandwidth.
The authors of \cite{lu2015network} applied network coding to optimize the reliability performance of frequently accessed data in cloud storage systems.

\subsection{Network Coding for Data Security}
Another research area of network coding is to prevent data being eavesdropped during transmission. The information-theoretical security problem for an untrusted channel was first discussed in \cite{Ozarow1985}. A network coding system was built to prove that a wiretapper cannot obtain any information from the transmitted message \cite{Cai2002}. A weaker type of security issue was investigated in\cite{Bhattad2005}, where a node can decode packets only after receiving sufficient linear independent encoded data.
The construction of a secure linear network code for a wiretap network was presented in \cite{cai2011secure}.

The secrecy capacity for a network-coded cloud storage system was investigated in  \cite{pawar2010secure,shah2011information}, where  the secrecy capacity is defined as the maximum amount of data that can be securely stored under the perfect secrecy condition.
The perfect secrecy condition ensure the eavesdropper cannot obtain any information of source data.
The secrecy capacity for nodes with different storage capacities was derived in \cite{ernvall2013capacity}. The coding scheme that can achieve the storage upper bound of secrecy capacity was proposed in \cite{rawat2014secure}. The maximum data size being stored under the perfect secrecy condition for any number of eavesdropped  nodes was determined in \cite{goparaju2013data}. The authors of \cite{Shah2011} considered how to achieve the  information-theoretical secrecy when an eavesdropper can access some data  in the storage nodes.

For secure storage over multiple clouds, similar to this work, the authors of \cite{oliveira2012coding} proposed a security protection scheme to ensure that no symbols can be decoded by an eavesdropper, which is weaker than perfect secrecy. In \cite{Chen2016EavesdroppingPF}, a link eavesdropping problem in a network-coded cloud storage system was investigated.
A publicly verifiable protocol for network coded cloud storage was proposed in \cite{chen2016secure}.

\subsection{Objective of This Paper}
Different from the previous works focusing on the security and reliability enhancement of network coding, in this paper we focus on the storage efficiency  and perfect security when applying network coding in multiple untrusted clouds.
We define the overflow problem when network coding is applied in a cloud storage system, which has not been discussed previously. The overflow problem will result in extra extended encrypted data in the format of digits during encoding process, thereby increasing storage and computation cost.
To overcome the overflow problem, we develop a systematic design methodology for calculating the encoding and storage parameters of a network-coded cloud storage system. Based on the proposed method, we further solve the storage cost optimization problem under the perfect secrecy constraint. The ultimate goal of this paper is to demonstrate that 
 the performance of a network-coded cloud storage system can be improved by jointly designing the encoding and storage and parameters.

\section{System Model and Problem Statement}  \label{model}
Now we describe the coding scheme adopted in this paper and give the formal definition of the  overflow problem.

\subsection{Coding Scheme}
Consider the original data vector ${\bf{b}} = {({b_1}, \ldots ,{b_n})^T}$ with base ${d}$, where elements ${b_i}$ are independent random integers uniformly distributed over ${\left\{ {0, \ldots ,d - 1} \right\}}$.
We use the terms original data and plaintext data interchangeably in this paper.
The goal of a cloud user is to securely store ${\bf{b}}$ to multiple cloud databases.
To achieve this goal, we adopt the same network coding scheme as \cite{oliveira2012coding}, in which the input data are mapped to encoded symbols by linear transformation.

Denote ${\bf{A}}$ as an $n \times n$ Vandermonde matrix, where $\left[ {{A_{i,j}}} \right] = (a_j^{i - 1})$.
${\bf{A}}$ is used for the encoding matrix, where all the coefficients ${a_i}$ are distinct nonzero elements over a finite field ${F_q}$, $q = {2^k} > n$. Note that ${\bf{A}}$ can be a $(n+m) \times n$ matrix, where the the amount of redundancy $m$ depends on the reliability requirement of the storage system.

A cloud user encodes data ${\bf{c}} = {({c_1}, \ldots ,{c_n})^T} = {\bf{Ab}}$ and splits the encoded data into ${p}$ parts. We assume the cloud user can arbitrarily store any piece of the encoded data to any cloud database. Let ${{{\bf{\tilde c}}_i}( i = 1, \ldots ,p)}$ be the encoded data vector stored in the ${i}$-th cloud database. A legitimate user can collect ${{{\bf{\tilde c}}_i}}$ from the cloud databases and obtain the original data by performing ${{{\bf{A}}^{{\bf{ - 1}}}{\bf{c}}}}$.

\subsection{Security Model}
Assume that an eavesdropper has infinite computing power, but can access only one cloud database.  Also, it is assumed that the eavesdropper can have the full information about the encoding and decoding schemes, including the knowledge of the encoding matrix. The objective of an eavesdropper is to guess the original data.
Although we consider only one eavesdropper in this paper, our result can be extended to the case of multiple eavesdroppers.

In our considered cloud storage system, every cloud database can support different security levels \cite{Barua2011}. Denote ${P_{{e_i}}}$ as the probability that the ${i}$-th cloud database can resist attacks. Also, the cloud user specifies a security requirement ${P_u}$, which represents the maximum probability that an eavesdropper can guess the original data.
Next, we will show how to solve the overflow problem subject to the constraints of  the security 
requirement when considering  distributing encoded symbols to multiple cloud databases.

\subsection{Overflow Problem}
Although it was proved that the aforementioned network coding scheme can help prevent eavesdroppers from obtaining the information of the original data \cite{Oliveira2010}, the overflow problem  occurs from the encoding process and the storage process. Specifically,
if the encoding parameter and the storage parameter are mismatched,  the length of  encoded data in digital format may become larger than the length of the original data in digital format.
As a result, storage spaces are wasted due to redundant encoded data.
Now we formally state this problem by introducing the following definition.

\begin{defi} [Strictly Non-overflow] \label{non-overflow}  Let ${{l_d}(a)}$  be the number of digits that represents ${a}$ in base ${d}$. A piece of encoded data ${\bf{c}} = {({c_1}, \ldots ,{c_n})^T}$ is considered to be strictly non-overflow if and only if ${{l_d}({c_i}) \le {l_d}({b_i})}$ for every $i$. Thus,  the length of the encoded data is equal to that of the plaintext data.
\end{defi}

\begin{table*}
\vspace*{0cm} \caption{\bfseries Example of the definitions for overflow problem} \label{example}
\begin{center}
\vspace*{-0.1cm}\scriptsize
\renewcommand{\arraystretch}{2}
\begin{tabular}{|l|l|l|l|l|l|l|l|l|} \cline{1-8}

\multicolumn{1}{|c}{}
&\multicolumn{1}{||c}{${\bf{A}}$}
& \multicolumn{1}{|c}{${\bf{b}}$}
&\multicolumn{1}{|c}{${\bf{c}}$}
&\multicolumn{1}{|c}{${{\bf{\tilde c}}_1} = ({c_1},{c_3})$}
&\multicolumn{1}{|c}{${{\bf{\tilde c}}_2} = ({c_2})$}
&\multicolumn{1}{|c}{strictly non-overflow}
&\multicolumn{1}{|c|}{$3$-bounded non-overflow} \\
\hline
%
%
\multicolumn{1}{|c} {Case1}
&\multicolumn {1}{||c} {$\left[ {\begin{array}{*{20}{c}}
1&1&1\\
5&1&2\\
6&1&4
\end{array}} \right]$}
&\multicolumn {1}{|c} {$\left( {0,{\rm{1}},0} \right)^T$}
&\multicolumn{1}{|c}  {$\left( {{\rm{1}},{\rm{1}},{\rm{1}}} \right)^T$}
&\multicolumn{1}{|c}  {$\left( {{\rm{1}},1} \right)$}
&\multicolumn{1}{|c}  {$\left( {{\rm{1}}} \right)$}
&\multicolumn{1}{|c}  {Yes}
&\multicolumn{1}{|c|}  {Yes}   \\
\hline
%
%
\multicolumn{1}{|c} {Case2}
&\multicolumn {1}{||c} {$\left[ {\begin{array}{*{20}{c}}
1&1&1\\
5&1&2\\
6&1&4
\end{array}} \right]$}
&\multicolumn {1}{|c} {$\left( {0,{\rm{1}},{\rm{1}}} \right)^T$}
&\multicolumn{1}{|c}  {$\left( {{\rm{10}},{\rm{11}},{\rm{101}}} \right)^T$}
&\multicolumn{1}{|c}  {$\left( {{\rm{10}},{\rm{101}}} \right)$}
&\multicolumn{1}{|c}  {$\left( {{\rm{11}}} \right)$}
&\multicolumn{1}{|c}  {No}
&\multicolumn{1}{|c|}  {Yes}   \\
\hline
\end{tabular}
\end{center}
\end{table*}

\begin{defi} [${\alpha }$-bounded Non-overflow] \label{bounded non-overflow} Let ${\left| {{{{\bf{\tilde c}}}_i}} \right|}$ denote the number of elements in  ${{{\bf{\tilde c}}_i}}$. A piece of encoded data ${\bf{c}} = {({c_1}, \ldots ,{c_n})^T}$  is considered to be \emph{${\alpha }$-bounded Non-overflow} if and only if
\begin{eqnarray}
\sum\limits_{j = 1}^{\left| {{{{\bf{\tilde c}}}_i}} \right|} {{l_d}} ({c_j}) \le \left| {{{{\bf{\tilde c}}}_i}} \right|\alpha {l_d}({b_i})
\enspace, \nonumber
\end{eqnarray}
for ${1 \le i \le p}$.
\end{defi}

Assume the encoded data are randomly stored in cloud databases. Hence, the increasing cost of storage or computation resources caused by data extension can be measured by the extension degree ${\alpha  = \frac{{{l_d}({c_i})}}{{{l_d}({b_i})}}}$  of the encoded data in the cloud database.
Table \ref{example} show the coding results for the two different overflow cases where $d = 2$ and $p = 2$. In case 1 there are no redundant digits after the encoding process, but in case 2 the extension degree is bounded by $3$.

\section{Network Coding based Secure Storage (NCSS) Scheme }  \label{scheme}
In this section, we analyze the overflow problem of a network-coded cloud storage system.
We first give the criteria  to choose the proper length of the data element to be encoded.
Next, we present the data distribution method for achieving the required security level. Finally, we describe the system design methods of the NCSS scheme.
Table \ref{description} summarizes the notations used in this paper.

\begin{table}
\renewcommand{\arraystretch}{1.3}\tabcolsep=1ex
\caption{\bfseries Notations in this paper} \label{description}
\centering
\begin{tabular}{c c c cc}
\hline\hline
\bfseries Notations & \bfseries Descriptions \\
\hline
${\bf{b}}$ & Original data array \\
\hline
$d$ & Base of ${b_i}$ \\
\hline
${l_d}(a)$   & Number of digits  that represents  $a$ in base $d$ \\
\hline
${\bf{A}}$   & Encoding matrix \\
\hline
$k$ & Use Galois Field size ${2^k}$ for ${\bf{A}}$ \\
\hline
$n$& Matrix size of  ${\bf{A}}$ \\
\hline
$p$  & Total number of cloud databases \\
\hline
${\bf{c}}$   &Encoded data vector  \\
\hline
${{\bf{\tilde c}}_i}$   &Encoded data vector  that stored in the ${i}$-th cloud database\\
\hline
$\left| {{{{\bf{\tilde c}}}_i}} \right|$ & Number of elements in ${{\bf{\tilde c}}_i}$  \\
\hline
${s_i}$ & Number of digits in ${b_i}$ \\
\hline
${\bf{b}}'$ & Regrouped data array \\
\hline
$r$ & Size of  ${\bf{b}}'$\\
\hline
${P_{{e_i}}}$&  Probability of the ${i}$-th cloud database can resist attacks \\
\hline
${P_g}$ & Probability that an eavesdropper can guess the original data\\
\hline
${P_u}$ & Security requirement: Maximum probability that an eavesdropper can guess the original data \\
\hline\hline
\end{tabular}
\end{table}

\subsection{Coding Analysis}
The following theorem can help calculate the encoding parameters to avoid the overflow problem of the secure network coding storage system.

\newtheorem{theorem}{Theorem}
\begin{theorem}  \label{theorem1}
Let ${s_i}$ be the number of digits in ${b_i}$. Then, the system is strictly non-overflow if ${s_i} = s = \frac{k}{{{{\log }_2}d}}\enspace.$
\end{theorem}
\begin{proof}
First, we assume that ${s_i} < \frac{k}{{{{\log }_2}d}}$. Then, we have
\begin{eqnarray}
\frac{k}{{{{\log }_2}d}} = k  {\log _d}2 = {\log _d}{2^k}
\enspace. \label{s_log}
\end{eqnarray}
Because the coding process is manipulated with integers, we have ${s_i} \le {\log _d}({2^k} - 1)$. Since ${c_i}$ is distributed over $\left\{ {0, \ldots ,{2^k} - 1} \right\}$, the maximum number of digits used to represent an encoded element is ${l_d}{({c_i})_{\max }} = {\log _d}({2^k} - 1)$. Furthermore, the number of digits in ${b_i}$ can be represented as ${l_d}({b_i})$. Thus, we have
\begin{eqnarray}
 {s_i} = {l_d}({b_i}) \le {\log _d}({2^k} - 1) = {l_d}{({c_i})_{\max }} \enspace.
\end{eqnarray}
As a result, the overflow problem occurs because the length of encoded data may be larger than the length of the original data. Secondly, we assume that ${s_i} > \frac{k}{{{{\log }_2}d}}$. We take exponentiation with base $d$ on both sides and  we have ${d^{{s_i}}} > {d^{{{\log }_d}{2^k}}} = {2^k}$ from (\ref{s_log}). Since ${b_i} = {d^{{s_i}}}$,  it contradicts the fact that the maximum value of ${b_i}$  is ${2^k} - 1$. Hence, it follows that ${s_i} = s = \frac{k}{{{{\log }_2}d}}$.
\end{proof}
\begin{theorem}   \label{theorem2}
The system is ${\alpha }$-bounded non-overflow if  ${s_i} \ge \frac{1}{\alpha }{\log _d}({2^k} - 1)$  for every $i$.
\end{theorem}

\begin{proof}
Since ${s_i} = {l_d}({b_i})$ and ${l_d}{({c_i})_{\max }} = {\log _d}({2^k} - 1)$, we have
\begin{align}
\sum\limits_{j = 1}^{\left| {{{{\bf{\tilde c}}}_i}} \right|} {{l_d}({c_j})}
\le& \left| {{{{\bf{\tilde c}}}_i}} \right|  {\log _d}({2^k} - 1)\nonumber \\
 =& \alpha  \cdot \frac{1}{\alpha }  \left| {{{{\bf{\tilde c}}}_i}} \right| {\log _d}({2^k} - 1)\nonumber \\
 \le& \alpha   \left| {{{{\bf{\tilde c}}}_i}} \right|  {s_i}\nonumber \\
 = &\alpha   \left| {{{{\bf{\tilde c}}}_i}} \right|  {l_d}({b_i})
 \enspace.
\end{align}
\end{proof}

Theorem \ref{theorem1} and \ref{theorem2} give the criteria of selecting the length of the plaintext data element. Next, we relate the security requirement to the amount of encoded stored data.

\begin{theorem}   \label{theorem3}
The system satisfies the security requirement ${P_u}$ if
\begin{eqnarray}
\sum\limits_{j = 1}^{\left| {{{{\bf{\tilde c}}}_i}} \right|} {{l_d}({{\tilde c}_i}(j)) \le \sum\limits_{t = 1}^n {{l_d}({c_t})}  + {{\log }_d}{P_u} - } {\log _d}{\rm{(1 - }}{P_{{e_i}}}) \enspace, \nonumber
\end{eqnarray}
for  $1 \le i \le p$.
\end{theorem}

\begin{proof}
Recall that an eavesdropper can access only one cloud database. Hence, the probability ${P_g}$ that an eavesdropper can guess the original data is the product of the invasion probability of the cloud database and the probability of
guessing the remaining encoded digits. It follows that
\begin{align}
{P_g} &{\rm{ =  (1 - }}{P_{{e_i}}})  {d^{ - \left( {\left. {\sum\limits_{t = 1}^n {{l_d}({c_t})}  - \sum\limits_{j = 1}^{\left| {{{{\bf{\tilde c}}}_i}} \right|} {{l_d}({{\tilde c}_i}(j))} } \right)} \right.}}\nonumber \\
&\le {\rm{(1 - }}{P_{{e_i}}})  {d^{{{\log }_d}{P_u} - {{\log }_d}{\rm{(1 - }}{P_{{e_i}}})}}\nonumber \\
&= {P_u} \enspace. \label{theorem3_proof}
\end{align}
\end{proof}

\subsection{System Design}

\begin{figure}[t]
\centering
\includegraphics[width=0.8\textwidth]{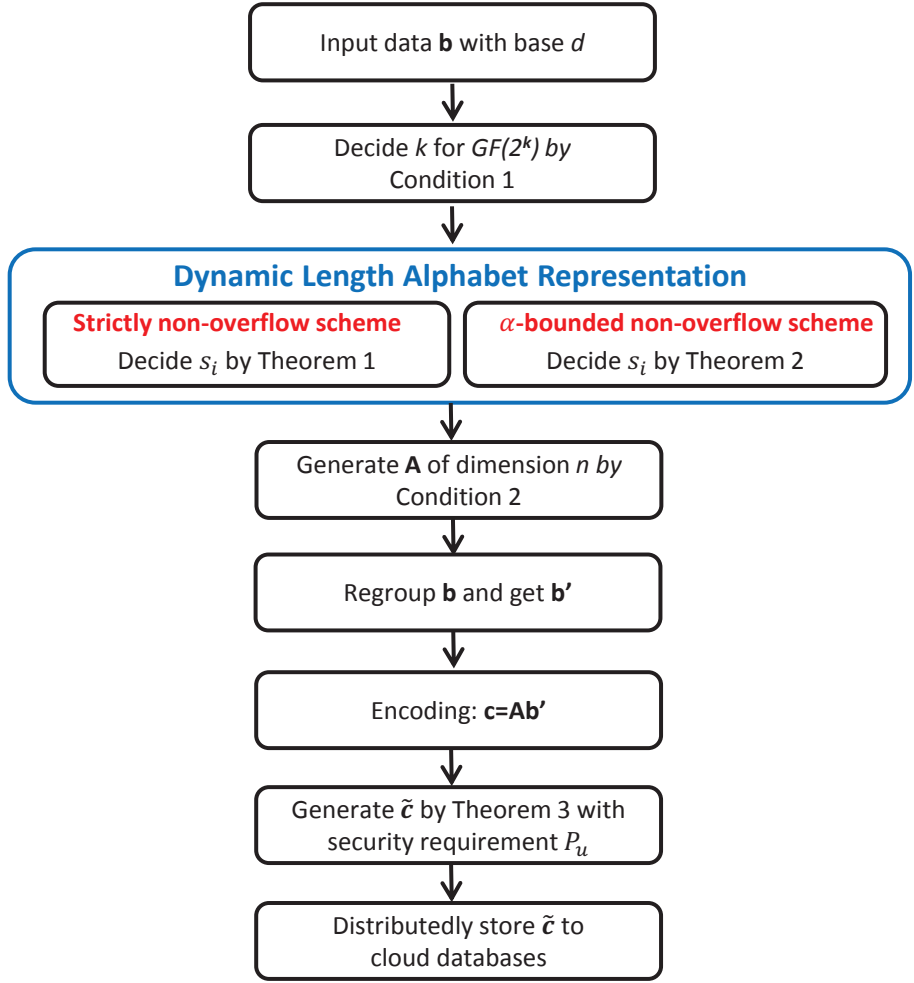}\\
\caption{ System flow of NCSS scheme.} \label{fig3}
\end{figure}

The proposed NCSS scheme can be divided into three steps.
First, a dynamic-length alphabet representation of network coded data is adopted based on Theorem \ref{theorem1} and Theorem \ref{theorem2}.
Second, the original data are  preprocessed and regrouped before the encoding process.
Third, the regrouped data are encoded and distributed to the corresponding cloud databases.

Figure \ref{fig3} shows the system flow of the proposed NCSS scheme. Assume that a cloud user wants to store a single-digit data array ${\bf{b}} = {({b_1}, \ldots ,{b_m})^T}$ with base $d$ to the $p$ cloud databases. We first choose a power $k$ for the field characteristics according to the following condition.
\newtheorem{con}{Condition}
\begin{con}   \label{con1}
${2^k} \ge d$
\end{con}
The field size must be larger than the maximal value of the data array element $d - 1$. Otherwise,  some data elements cannot be represented in the field.
After that, a proper length of data elements ${s_i}$ can be decided according to Theorem \ref{theorem1} and Theorem \ref{theorem2}. This step is called dynamic length alphabet representation. We then regroup ${\bf{b}}$ to ${{\bf{b}}^\prime } = ({b_1}...{b_{{s_1}}},{b_{{s_1} + 1}}...{b_{{s_1} + {s_2}}}, \cdots ,{b_{{{\hat s}_{r - 1}} + 1}}...{b_{{{\hat s}_r}}})$ based on the value of ${s_i}$, where ${{\hat s}_r} \buildrel \Delta \over = \sum\limits_{i = 1}^r {{s_i}}$. Next, we generate an $n \times n$ encoding matrix ${\bf{A}}$ with the following condition.
\begin{con}   \label{con2}
$n < {2^k}$ and $n \le r$
\end{con}
Since matrix ${\bf{A}}$ is constructed from $n$ distinct elements over the Galois Field, we have $n < {2^k}$. In addition, the matrix multiplication cannot be operated if the size of encoding matrix is larger than the size of regrouped data array. We then encode ${\bf{b}}'$ with ${\bf{A}}$ and obtain the encoded data array ${\bf{c}} = {({c_1}, \ldots ,{c_n})^T}$. Finally, ${\bf{c}}$ can be regrouped to ${\bf{\tilde c}}$ by Theorem \ref{theorem3}, which specifies the maximum amount of encoded data that can be stored in a cloud database according to user's security requirement. Finally, the elements of ${\bf{\tilde c}}$ are distributed to the corresponding $p$ cloud databases.

Table \ref{example2} shows an example of the proposed NCSS scheme in the strictly non-overflow case.
We assume that the original data is ${\bf{b}}= \left( {0,0,{\rm{1}},0,{\rm{1}},{\rm{1}},{\rm{1}},0,{\rm{1}}} \right)$ and the encoded data are stored to two cloud databases with ${P_{{e_1}}} = 0.5$, ${P_{{e_2}}} = 0.25$, and ${P_u} = \frac{1}{{64}}$. From Theorem \ref{theorem3},  the maximal numbers of digits that can be stored in the first and the second cloud database are $4$ and $5$, respectively.

\begin{table*}
\vspace*{0cm} \caption{\bfseries Example of adopting NCSS scheme in storing encoded data to two cloud databases}
\begin{center}
\vspace*{-0.1cm}\scriptsize
\renewcommand{\arraystretch}{2}
\begin{tabular}{|l|l|l|l|l|l|l|l||l||l||l|} \cline{1-10}

\multicolumn{1}{|c}{${\bf{b}}$}  \label{example2}
&\multicolumn{1}{|c}{$d$}
&\multicolumn{1}{|c}{$k$}
&\multicolumn{1}{|c}{$s$}
&\multicolumn{1}{|c}{${\bf{b}}'$}
&\multicolumn{1}{|c}{$r$}
&\multicolumn{1}{|c}{$n$}
&\multicolumn{1}{|c}{${\bf{A}}$}
&\multicolumn{1}{|c}{${\bf{c}}$}
&\multicolumn{1}{|c|}{${\bf{\tilde c}}$} \\
\hline
%
%
\multicolumn{1}{|c}{$\left( {0,0,{\rm{1}},0,{\rm{1}},{\rm{1}},{\rm{1}},0,{\rm{1}}} \right)$}
&\multicolumn{1}{|c}{$2$}
&\multicolumn{1}{|c}{$3$}
&\multicolumn{1}{|c}{$3$}
&\multicolumn{1}{|c}{$\left( {00{\rm{1}},0{\rm{11}},{\rm{1}}0{\rm{1}}} \right)$}
&\multicolumn{1}{|c}{$3$}
&\multicolumn{1}{|c}{$3$}
&\multicolumn{1}{|c}{$\left[ {\begin{array}{*{20}{c}}
1&1&1\\
5&1&2\\
6&1&4
\end{array}} \right]$}
&\multicolumn{1}{|c}{$\left( {{\rm{010}},1{\rm{00}},{\rm{00}}1} \right)$}
&\multicolumn{1}{|c|}{$\left( {{\rm{010}}1,{\rm{0000}}1} \right)$} \\
\hline

\end{tabular}
\end{center}
\end{table*}

\section{Security Analysis}  \label{analysis}

In this section, we analyze the proposed NCSS scheme in terms of security level and storage cost. First, we discuss the issue of enhancing security level from a system design aspect. 
Then we derive the upper bound of data size that can be stored in the cloud under the constraint of perfect secrecy. 

To begin with,  from (\ref{theorem3_proof})  we know that the lower bound of the security requirement ${P_u}$ is 
\begin{align}
{\rm{(1 - }}{P_{{e_i}}})  {d^{ - \left( {\left. {\sum\limits_{t = 1}^n {{l_d}({c_t})}  - \sum\limits_{j = 1}^{\left| {{{{\bf{\tilde C}}}_i}} \right|} {{l_d}({{\tilde c}_i}(j))} } \right)} \right.}} \le {P_u}
\enspace. \nonumber
\end{align}
Since ${l_d}({c_t})$ is proportional to the size of Galois Field,  it is anticipated that a larger size of encoding matrix $n$ and a large value of power $k$ of the field characteristics  can result in a higher security level for the network coding storage system.
However, increasing these encoding parameters can result extra coding complexity.
Next, we show that the security level can be enhanced to the perfect secrecy by storing a certain amount of encoded data in the local machine.
The notion of perfect secrecy represents that an eavesdropper can get no information of the original message.
\begin{defi} [Perfect Secrecy Criterion \cite{massey1988introduction}] \label{defi_perfect_secrecy}
Let $S$  denote the random variable associated with the secret data fragments and $E$ denote the random variable associated encoded fragments observed by the eavesdropper.
The perfect secrecy requires \begin{align}
H({\rm{S}}|E) = H({\rm{S}})
\enspace, \nonumber
\end{align}
where H({\rm{X}}) represents the entropy of a random variable $X$.
\end{defi}

In the worst case,  an eavesdropper can access the encoded data of all the cloud databases.  The following theorem can be applied to specify the maximal amount of encoded data fragments that can be stored in the cloud, while keeping the rest of data in a local machine to ensure perfect secrecy.
\begin{theorem}   \label{theorem4}
Assume that $w$-digit secret information is encoded with $n-w$-digit data ${\bf{b}}$. For both strictly non-overflow and ${\alpha }$-bounded non-overflow schemes, a cloud user can store at most $\sum\limits_{j = 1}^n {{l_d}} ({c_j}) - w$ digits of encoded data to the cloud under the perfect secrecy criterion.
\end{theorem}

\begin{proof}
Let $\mathbf{e}^{(h)}$ represent  a subset containing any $h$ components of vector $\mathbf{e}$. We use $\mathbf{e}_{i:j}$ to denote the subvector formed from the $i$-th to the $j$-th position of vector $\mathbf{e}$. The set of rows from the $i$-th to the $j$-th position of matrix $\mathbf{D}$ is represented as $\mathbf{D}_{i:j}$. In addition, $b_{i}$ are independent random variables uniformly distributed over $\mathbf{F_{\emph{q}}}$ with entropy $H(b_{i}) = H(b)$.

For simplicity, without loss of generality, assume that  $t$ contiguous components of the encoded data $\mathbf{c}_{p+1:p+t}$ are stored to the clouds.
Then we can obtain
\begin{eqnarray}
H({{\bf{b}}^{(w)}}) &=& H({{\bf{b}}^{(w)}}|{{\bf{c}}_{p + 1:p + t}}) - H({{\bf{b}}^{(w)}}|{\bf{c}}) \label{new entropy_1} \\
&=&I({{\bf{b}}^{(w)}};{\bf{c}}) - I({{\bf{b}}^{(w)}};{{\bf{c}}_{p + 1:p + t}}) \label{new entropy_2} \\
&=&H({\bf{c}}) - H({{\bf{c}}_{p + 1:p + t}}) - H({\bf{c}}|{{\bf{b}}^{(w)}}) + H({{\bf{c}}_{p + 1:p + t}}|{{\bf{b}}^{(w)}}) \label{new entropy_3} \\
&\le& H({\bf{c}}) - H({{\bf{c}}_{p + 1:p + t}})
\enspace. \label{new entropy_final}
\end{eqnarray}
In the above equations, (\ref{new entropy_1}) holds because of the perfect secrecy criterion and due to  the fact that the secret information can be reconstructed  if the entire codewords are given.
In (\ref{new entropy_final}), we have $H({{\bf{c}}_{p + 1:p + t}}|{{\bf{b}}^{(w)}}) - H({\bf{c}}|{{\bf{b}}^{(w)}}) \le 0$ since
\begin{eqnarray}
H({\bf{c}}|{{\bf{b}}^{(w)}})- H({{\bf{c}}_{p + 1:p + t}}|{{\bf{b}}^{(w)}}) =H({{\bf{c}}_{p + t + 1:n}}|{{\bf{b}}^{(w)}},{{\bf{c}}_{p + 1:p + t}})\enspace.  \nonumber 
\end{eqnarray}
Since $b_{i}$ are i.i.d random variables, it follows that
\begin{eqnarray}
H({{\mathbf{b}^{\left( w \right)}}})&=& H\left( {{b_{seq(1)}},{b_{seq(2)}}, \ldots ,{b_{seq(w)}}} \right) \nonumber \\
&=& wH(b)  \enspace, \label{w_entropy}
\end{eqnarray}
where $seq(j)$ is the $j$-th element of a random integer sequence within the range $1$ to $n$.
Because the encoded data vector $\mathbf{c}$ contains the entire information of $\mathbf{b}$ at most, we can obtain
\begin{align}
H({\bf{c}}) \le nH(b)
\enspace. \label{c_entropy}
\end{align}
Moreover, the $n\times n$ Vandermonde matrix $\mathbf{A}$ is nonsingular \cite{Klinger_21}. Thus the eavesdropper can apply the Gaussian elimination to obtain the  reduced row echelon form of the submatrix $\mathbf{S}$, whose elements are $[S_{i,j}] = [A_{i,j}]$ for $p+1 \leq i,j \leq p+t$. The \textit{Eavesdropper Reduced Matrix} $\mathbf{M}$ can be obtained as 
\begin{align}
{{\rm{\mathbf{M}}}_{{\rm{p}} + 1:{\rm{p}} + {\rm{t}}}} = \left[ {\begin{array}{*{20}{c}}
{m_1^p}\\
 \vdots \\
{m_1^{p + t - 1}}
\end{array}\begin{array}{*{20}{c}}
{...}\\
 \vdots \\
{...}
\end{array}\begin{array}{*{20}{c}}
{\rm{|}}\\
{\rm{|}}\\
{\rm{|}}
\end{array}\begin{array}{*{20}{c}}
{}\\
{{{\bf{ I}}_{t}}}\\
{}
\end{array}\begin{array}{*{20}{c}}
{\rm{|}}\\
{\rm{|}}\\
{\rm{|}}
\end{array}\begin{array}{*{20}{c}}
{...}\\
 \vdots \\
{...}
\end{array}\begin{array}{*{20}{c}}
{m_n^p}\\
 \vdots \\
{m_n^{p + t - 1}}
\end{array}} \right] \enspace,
\end{align}
where the other element of $\mathbf{M}$ are the same as $\mathbf{A}$.
Hence, the eavesdropper have $t$ equations to solve $n$ unknown elements.
It implies that
\begin{align}
H({{\bf{c}}_{p + 1:p + t}}) = tH(b)
\enspace. \label{t_entropy}
\end{align}
Substituting (\ref{w_entropy}), (\ref{t_entropy}) and (\ref{c_entropy}) into (\ref{new entropy_final}), we obtain
\begin{align}
tH(b) \le nH(b) - {\rm{wH(b)}}
\enspace. \label{result}
\end{align}

The above equation shows that we can store at most the $n-w$ components of encoded data to the clouds under perfect secrecy criterion. For the strictly non-overflow scheme, we have only one digit in each component of encoded data. Thus, we can store at most $\sum\limits_{j = 1}^n {{l_d}} ({c_j})-w$ digits of encoded data to the clouds, while keeping the remaining $w$ digits in the local machines.
However, we may have multiple digits in each component of encoded data for ${\alpha}$-bounded non-overflow scheme. Let ${\mathbf{e}^{(\tilde h)}}$ represent a subset containing any $w$ fragmentary components of vector $\mathbf{e}$. With at least $n$ unknown digits,   knowing ${\mathbf{c}^{(\tilde w)}}$  cannot help solve $\mathbf{b}$. As a result, it follows that
\begin{eqnarray}
I\left( {{\mathbf{c}^{(\tilde w)}};\mathbf{b}} \right) = 0
\enspace.
\end{eqnarray}
Not that we still have $t$ equations to solve $n$ unknown elements. That is, 
\begin{eqnarray}
H(\mathbf{b}^{(w)}|\mathbf{c}_{p+1:p+t},{\mathbf{c}^{(\tilde w)}}) =H(\mathbf{b}^{(w)}|\mathbf{c}_{p+1:p+t})
\enspace.
\end{eqnarray}
Finally, we obtain
\begin{eqnarray}
I\left( {{\mathbf{c}_{p + 1:p + t}},{\mathbf{c}^{(\tilde w)}};{\mathbf{b}^{\left( w \right)}}} \right) =I\left( {{\mathbf{c}_{p + 1:p + t}};{\mathbf{b}^{\left( w \right)}}} \right)
\enspace.
\end{eqnarray}
Consequently, we can select $w$ digits of encoded data from different $w$ components, i.e., select one digit for each component.  These $w$-digit encoded data can be stored in the local machines, while  the remaining $\sum\limits_{j = 1}^n {{l_d}} ({c_j}) - w$ digits are stored to the clouds.
\end{proof}

\section{Storage Minimization} 
We are motivated to analyze the amount of stored encrypted data with the security requirement in terms of the probability that an eavesdropper can correctly guess the original data. This is 
because  only a certain amount of encoded data fragments can be stored in the local machines to enhance the security level, as shown in the previous section.
As the required security level increases, the amount of encoded data stored at the local site should increase.

\subsection{Solving Storage Minimization Problem}
Consider a cloud user keeps encoded data with length $l$ in each encoding operation and stores the remaining encoded data to $p$ cloud databases as shown in Fig.~\ref{costModel}. 
We assume all the cloud databases have the same capability of preventing attacks (i.e., ${P_{{e_i}}}$=${P_{{e}}}$) and the security requirement is ${P_{{u}}}$, which specifies the maximum probability that an eavesdropper can guess the original message.
In addition to the encoded data, the encoding matrix is stored at the local site.

The storage cost at the local site is the function of  encoding matrix size $n$ and the amount of encoded data stored at  a local machine for every encoding operation, denoted by $l$.
Let ${m}$ denote the length of the original message and $\alpha $ represent the number of encoding operations.
Subject to a given security requirement ${P_{{u}}}$, the storage cost minimization problem is expressed as
\begin{eqnarray}
  {\text{minimize    }}f(n,l) = {n^2}s + \alpha l \hfill \nonumber \\
  {\text{subject to  }}
  {(1 - {P_e})^p}  {d^{ - \alpha l}} \leqslant {P_u} \hfill \nonumber \\
  2 \leqslant n \leqslant {2^k} \hfill \nonumber \\
  l \leqslant n \hfill \nonumber \\
  \alpha ns = m \hfill \nonumber \\
  n,l \in {\mathbb{Z}^ + } \enspace, \label{original opt}
\end{eqnarray}
where $s$ is defined in Theorem 1. Note that an eavesdropper can guess the original message only if he/she can invade all the cloud databases and guess the encoded data in the local machine. It is observed that the optimization problem is nonconvex even if we relax the noncovex constraints $n,l \in {\mathbb{Z}^ + }$. The complete algorithm for solving this optimization problem is given in the Appendix.


\subsection{Discussions}
Figure \ref{fign_m} shows the optimal parameter setting for encoding matrix size ${n}$ versus the original message length ${m}$ for $d = 2$, ${P_e} = 0.5$, $p = 3$, and ${P_u} = {10^{ - 6}}$. As the message length increases, the size of the encoding matrix increases. A smaller encoding matrix size is preferred if  Galois field size is large.
 Due to the integer constraints in the optimization problem, the encoding matrix size increases in a step-like function.

Figure \ref{figf_m} shows  the storage cost $f(n,l)$ versus message length ${m}$ for $d = 2$, ${P_e} = 0.5$, and $p = 3$. Intuitively,  we need more storage space for lower ${P_u}$. 
However,  the storage cost with various ${P_u}$ are the same when ${m}$ exceeds certain threshold.  This is because the considered system is in the case of lower bound cost (i.e., ${l}=1$). Noteworthily, a larger ${k}$ can yield a smaller lower bound cost when ${m > 1000}$. In a general setting $k \in [8,16]$ \cite{angelopoulos2011energy}.   For ${m < 1000}$ it is suggested that the value of $k$ is set to  ${k=8}$; otherwise,  ${k=16}$.

\begin{figure}
\centering
\includegraphics[width=0.9\textwidth]{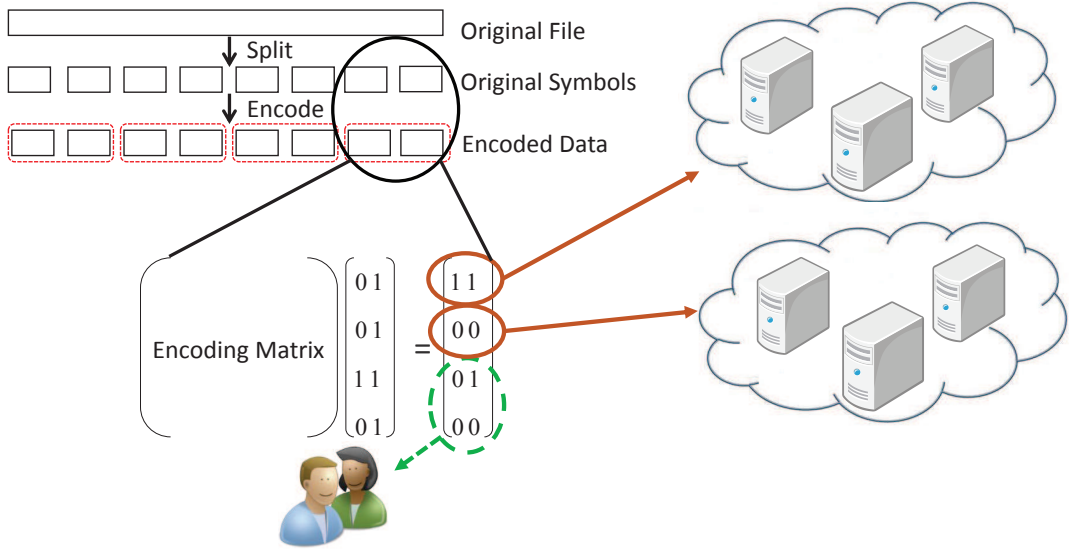}\\
\caption{Illustration of a user keeping a certain amount of encoded data at the local site to enhance security protection.} \label{costModel}
\end{figure}

\begin{figure}
\centering
\includegraphics[width=0.9\textwidth]{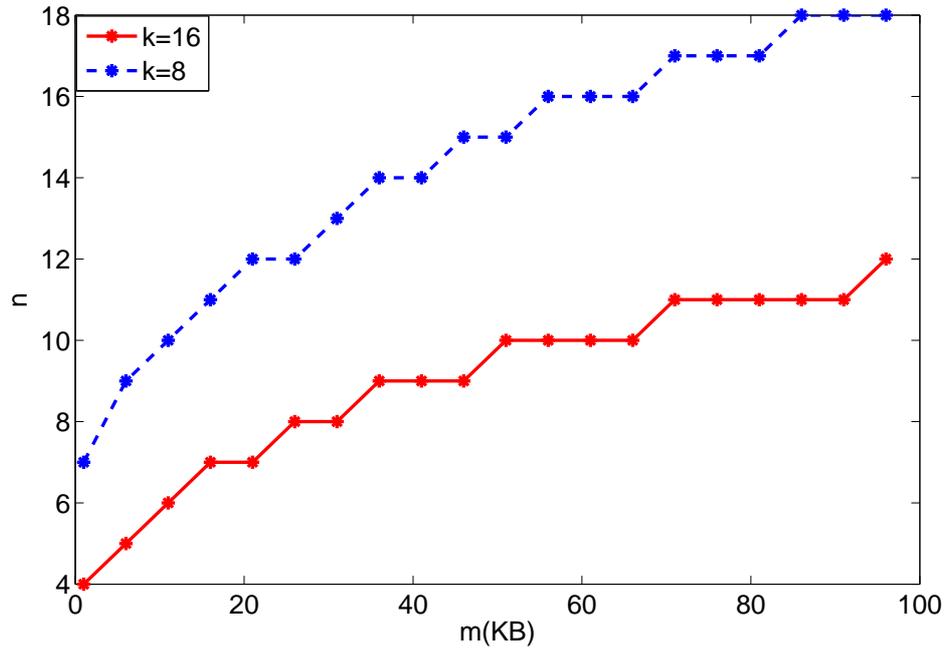}\\
\caption{Optimal parameter setting for encoding matrix size versus message length under different Galois Field sizes ${2^k}$.} \label{fign_m}
\end{figure}

\begin{figure}
\centering
\includegraphics[width=0.9\textwidth]{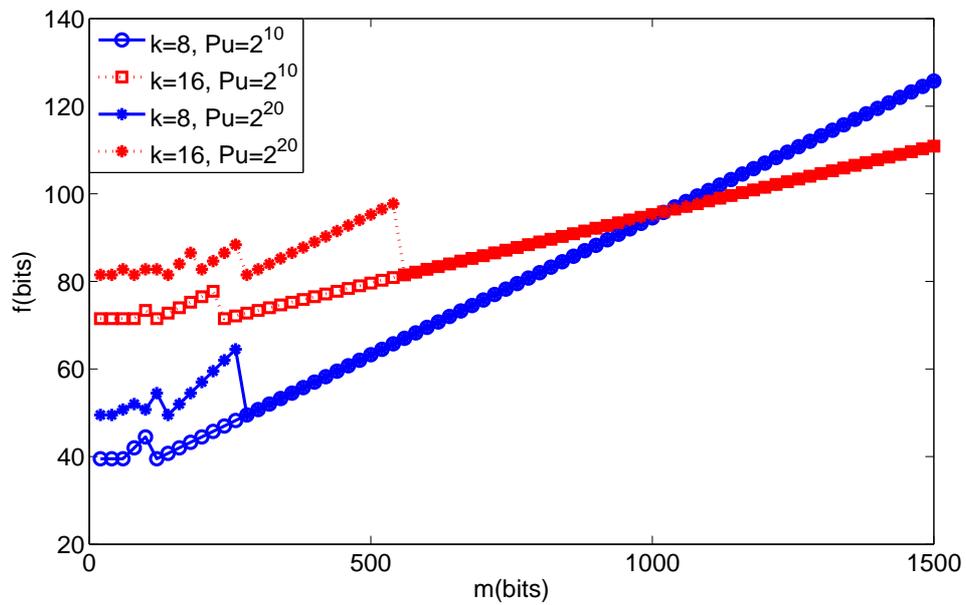}\\
\caption{Storage cost versus message length for different Galois Field sizes ${2^k}$ and security requirement ${P_u}$.} \label{figf_m}
\end{figure}

\section{Experimental Results}  \label{experiment}

Since the encoding process is performed on local machines, processing delay may be  performance bottlenecks.  Thus, it is of importance to investigate the impact of the system design parameters  on the delay performance when considering  a secure network coding 
storage scheme. We performed experiments on a commodity computer with an Intel Core i5 processor running at 2.4 GHz, 8 GB of RAM, and a 5400 RPM Hitachi 500 GB Serial ATA drive with an 8 MB buffer.

Figure \ref{fig4} shows the multiplication processing time of the
network coding storage system with  different sizes of  Galois Field. 
Although the complexity for the network coding is $O({n^2})$ modular multiplication, 
our result shows that the field size has only slight impact on the processing time, which supports our design methodology of selecting $k$.  Specifically, it indicates the possibility that   the security level can be enhanced significantly by selecting an appropriate design of $k$  but only pay a very small computational cost.

Figure \ref{fig5} shows the processing time between the strictly non-overflow and the ${\alpha }$-bounded non-overflow schemes for 2MB file with $p = 2$, where $\alpha  = 5$.
The processing time is longer for a smaller $n$ or $k$ since the numbers of encoding times  increase.
As a result, the system spends more time  in I/O operations and fetching data between the kernel and user \cite{Fragouli2006}.
Compared to the strictly non-overflow scheme, the ${\alpha }$-bounded non-overflow scheme requires more computation cost.
 The ${\alpha }$-bounded non-overflow scheme costs more than 11 times and 22 times of the processing time than that of the strictly non-overflow scheme when $k = 16$ and 8, respectively.
Finally,  the best performance is achieved when $n > 100$ for both non-overflow schemes. Because increasing $n$ results in a larger cost than  increasing $k$, we suggest to  fix $n = 100$ and adjust $k$ to meet the security requirements.

Figure \ref{fig6}  compares the processing time of  the strictly non-overflow and the ${\alpha}$-bounded non-overflow schemes versus the power of Galois Field characteristic $k$.
As shown in the figure,   the strictly non-overflow scheme is preferable to the ${\alpha}$-bounded  non-overflow scheme.
It is noteworthy that $k$ has negligible effect on the processing time of the strictly non-overflow scheme while it has a great impact on that of the ${\alpha }$-bounded non-overflow scheme.

\begin{figure}
\centering
\includegraphics[width=0.9\textwidth]{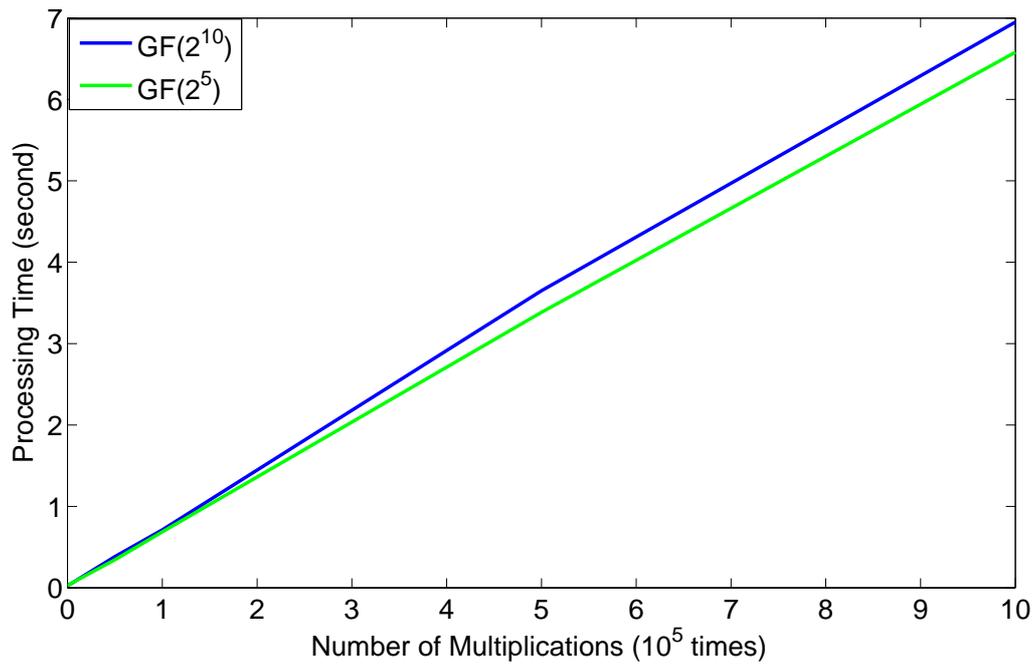}\\
\caption{Processing time versus the multiplication times for different Galois Fields ${2^k}$.} \label{fig4}
\end{figure}

\begin{figure}
\centering
\includegraphics[width=0.9\textwidth]{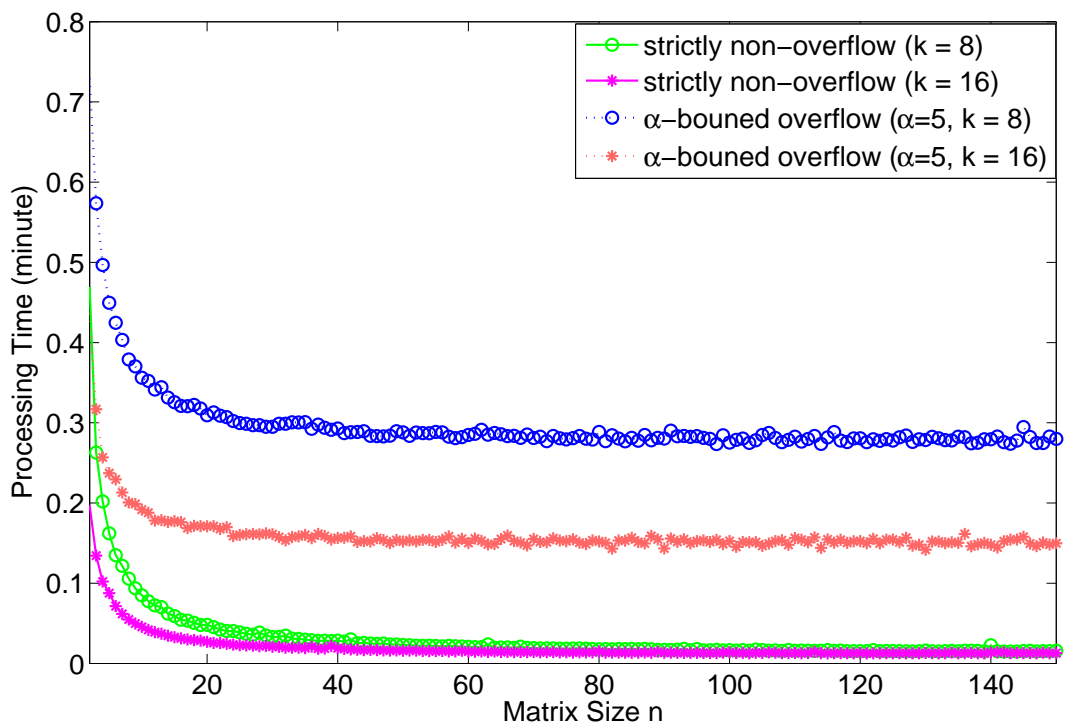}\\
\caption{Comparison of processing time between the strictly non-overflow and the ${\alpha }$-bounded non-overflow schemes versus matrix size $n$ with $p = 2$.} \label{fig5}
\end{figure}

\begin{figure}
\centering
\includegraphics[width=0.9\textwidth]{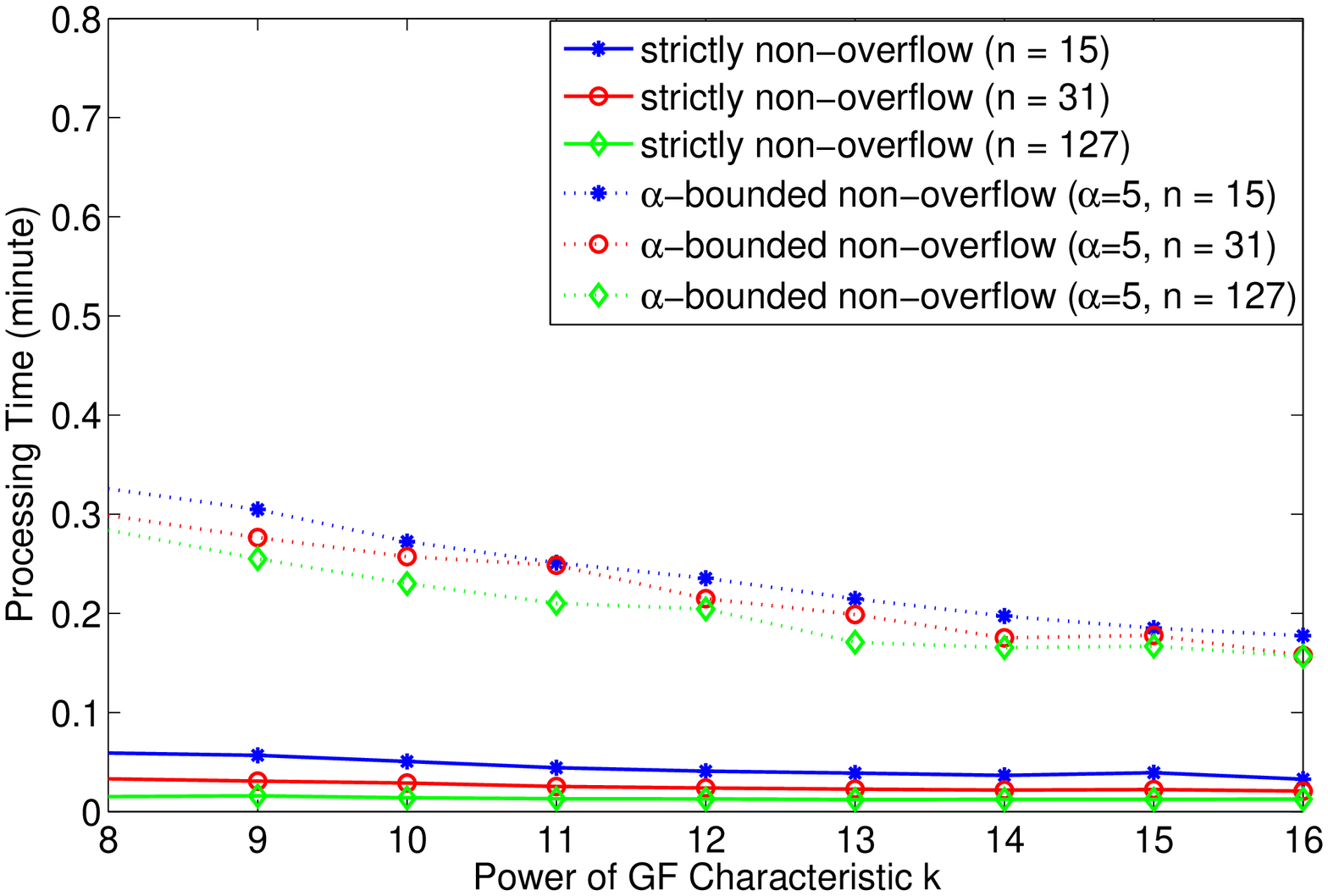}\\
\caption{Comparison of processing time between the strictly non-overflow and the ${\alpha }$-bounded non-overflow schemes versus power of Galois Field characteristic $k$ with $p = 2$.} \label{fig6}
\end{figure}

\section{Conclusions} \label{conclusion}
In this paper, we investigated the overflow problem in a network coding cloud storage system. When the overflow problem occurs, it does not only require more storage spaces but increases the processing time in encoding. We developed the network coding based secure storage (NCSS) scheme.  A systematic approach for the optimal encoding and storage parameters was provided to solve the overflow problem and minimize the storage cost. 
We also derived an analytical upper bound for the maximal allowable stored data in the cloud nodes under perfect secrecy criterion.
Our experimental results demonstrated  that encoding efficiency in terms of processing time
can be improved by jointly design of the encoding and the storage system parameters.
More importantly, we suggested the key design guidelines for secure network coding storage systems to optimize the performance tradeoff among security requirement, storage cost per node, and encoding processing time.
In the future research, we will further incorporate the factors of user budgets and file
recovery into the secure network coding distributed storage system.


\bibliography{reference}
\bibliographystyle{ieeetran}
\bibliographystyle{unsrt}

\section*{Appendix}

Here we first show that the original storage cost minimization (\ref{original opt}) is not convex even when the integer constraint is relaxed. Then we give the algorithm for solving the optimization problem by minimizing over separated variables.

\begin{theorem}  \label{theorem_nonconvex}
The objective function of the original storage cost minimization (\ref{original opt}) is not convex.
\end{theorem}

\begin{proof}
We consider the case of strictly non-overflow scheme. Substituting ${s_i} = s = \frac{k}{{{{\log }_2}d}}$ into (\ref{original opt}), the original storage cost minimization is equivalent to
\begin{eqnarray}
  {\text{minimize    }}\tilde f(n,l) =  \frac{k}{{{{\log }_2}d}}  {n^2} + \frac{{m{{\log }_2}d}}{k}  {n^{ - 1}}l \hfill \nonumber \\
  {\text{subject to  }}   - \frac{k}{{m{{\log }_2}d}}  {\log _d}\frac{{{P_u}}}{{{{(1 - {P_c})}^p}}}  n \leqslant l \leqslant n  \hfill \nonumber \\
   2 \leqslant n \leqslant {2^k}  \hfill \nonumber \\
   n,l \in {\mathbb{Z}^ + } \enspace. \label{reformulate opt}
\end{eqnarray}
Then we prove the theorem by showing that the Hessian matrix of the objective function is not positive semidefinite. The Hessian matrix of $\tilde f(n,l)$ is
\begin{eqnarray}
H(\tilde f) = \left[ {\begin{array}{*{20}{c}}
  {2a + 2bl{n^{ - 3}}}&{ - b{n^{ - 2}}} \\
  { - b{n^{ - 2}}}&0
\end{array}} \right]
\nonumber \enspace,
\end{eqnarray}
where $a = \frac{k}{{{{\log }_2}d}} > 0$ and $b = \frac{{m{{\log }_2}d}}{k} > 0$. Then, we solve the characteristic equation
\begin{eqnarray}
\det (H(\tilde f) - \lambda I) = {\lambda ^2} - (2a + 2bl{n^{ - 3}})\lambda  - {b^2}{n^{ - 4}} = 0
\nonumber \enspace.
\end{eqnarray}
We can obtain
\begin{eqnarray}
\lambda  = \frac{{2a + 2bl{n^{ - 3}} \pm \sqrt {{{\left( {2a + 2bl{n^{ - 3}}} \right)}^2}{\text{ +  }}4{b^2}{n^{ - 4}}} }}{2}
\nonumber \enspace.
\end{eqnarray}
Since the eigenvalues of $H(\tilde f)$ is not all positive, $H(\tilde f)$ is not positive semidefinite. Thus $\tilde f$ is not convex.

\end{proof}

We are now ready for solving the equivalent optimization problem (\ref{reformulate opt}) by minimizing over separated variables. Define ${{\tilde f}^ * }(n,l) \triangleq \mathop {\min }\limits_{n \in \mathbf{B},l \in \mathbf{A}} {\tilde f}$ and ${l^ * } \triangleq  \arg \mathop {\min }\limits_{l \in \mathbf{A}} \tilde f(n,l)$, where $\mathbf{A} = \{ x|x \in {\mathbb{Z}^ + },\frac{{nk}}{{m{{\log }_2}d}}  {\log _d}\frac{{{P_u}}}{{{{(1 - {P_c})}^p}}} \leqslant x < n\} $ and $\mathbf{B} = \{ x|x \in {\mathbb{Z}^ + },2 \leqslant x < {2^k}\}$. We first minimize over $n$
\begin{eqnarray}
{{\tilde f}^ * }(n,l) = \mathop {\min }\limits_{n \in B} \{ x|x = \mathop {\min }\limits_{l \in A} \tilde f(n,l)\}
\nonumber \enspace.
\end{eqnarray}
Since  $\mathop {\min }\limits_{l \in A} \tilde f(n,l)$ is a linear function with one variable in $\mathbb{Z}_{ +  + }^1$ for fixed $n$ and the coefficient is positive, we obtain
\begin{eqnarray}
{l^ * } = \min \{ \mathbf{A}\}  = \left\lceil {\frac{{ - nk}}{{m{{\log }_2}d}}  {{\log }_d}\frac{{{P_u}}}{{{{(1 - {P_c})}^p}}}} \right\rceil
\nonumber \enspace.
\end{eqnarray}
As a result, we can solve the optimization problem iteratively as:

\begin{itemize}

\item Step 0: Initiate $\mathbf{C} = \emptyset$ and $\mathbf{B} = \{ x|x \in {\mathbb{Z}^ + },2 \leqslant x < {2^k}\} $.
\item Step 1: Select  $n \in \mathbf{B} $ and set $l = \left\lceil {\frac{{ - nk}}{{m{{\log }_2}d}}  {{\log }_d}\frac{{{P_u}}}{{{{(1 - {P_c})}^p}}}} \right\rceil $.
\item Step 2: Calculate $c = \tilde f(n,l)$.
\item Step 3: Set $\mathbf{C} = \mathbf{C} \cup \{ c\}$ and $\mathbf{B} = \mathbf{B} - \{ n\}$.
\item Step 4: Iterate 1 to 4 until $ \mathbf{B} = \emptyset$.
\item Step 5: Obtain ${{\tilde f}^ * }(n,l) = \min \{ \mathbf{C}\}$.

\end{itemize}

\newpage
\vspace{-3cm}
\begin{biography}[{\includegraphics[width=1in,height
=1.1in,clip,keepaspectratio]{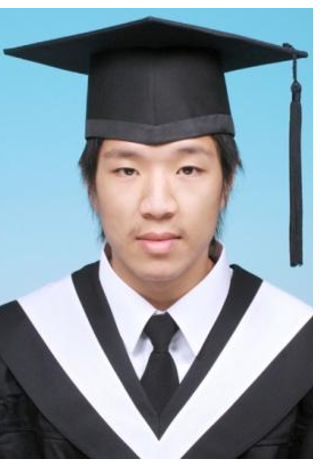}}]{Yu-Jia Chen}
received the B.S. degree and Ph.D. degree in electrical engineering from National Chiao Tung University, Taiwan, in 2010 and 2016, respectively. He is currently a postdoctoral fellow in National Chiao Tung University. His research interests include network coding for secure storage in cloud datacenters, software defined networks (SDN), and sensors-assisted applications for mobile cloud computing. Yu-Jia Chen has published 15 conference papers and 3 journal papers. He is holding two US patent and three ROC patent.
\end{biography}

\vspace{-3cm}
\begin{biography}[{\includegraphics[width=1in,height
=1.1in,clip,keepaspectratio]{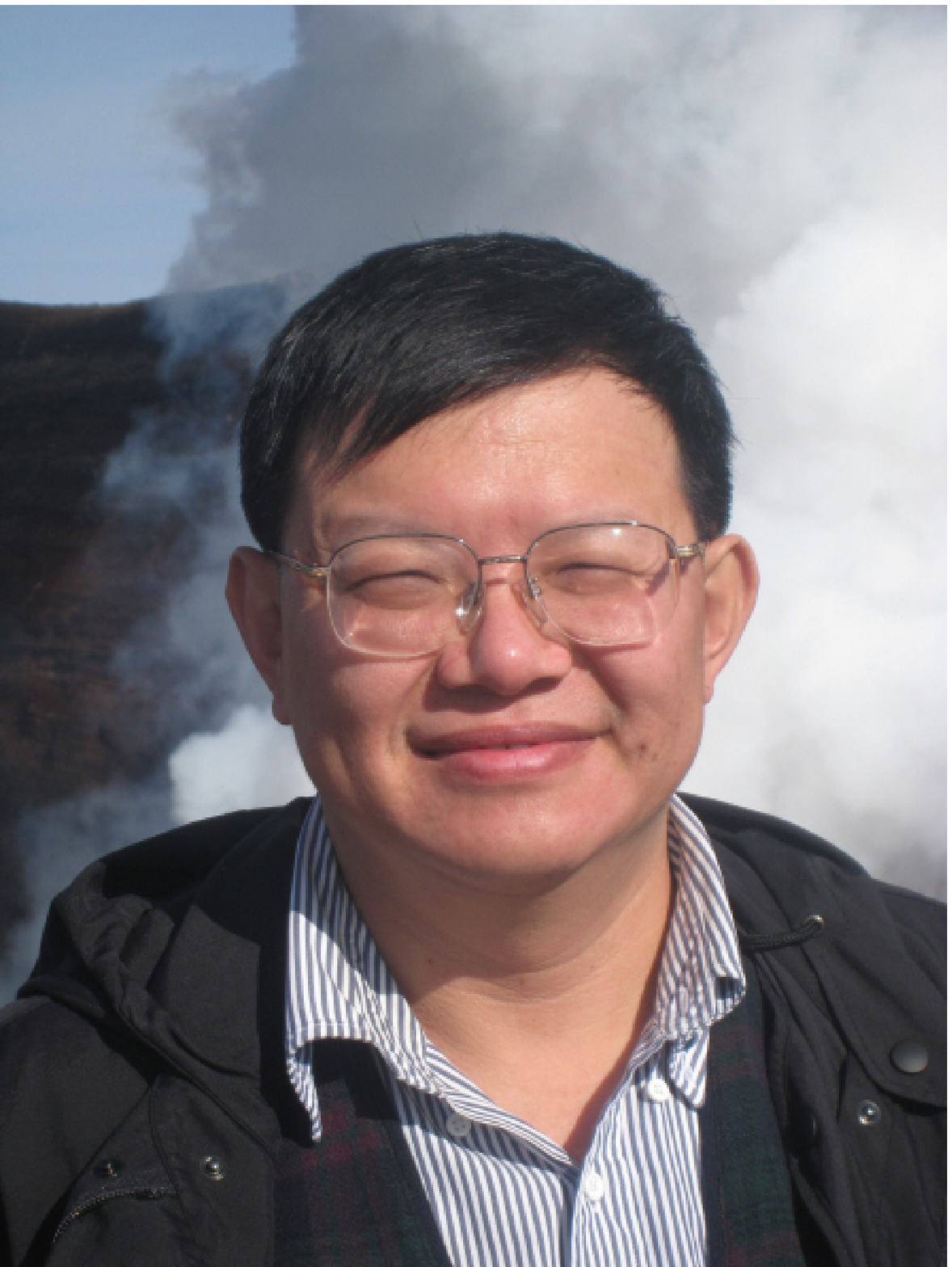}}]{Li-Chun Wang
(M'96 -- SM'06 -- F'11)} received the B.S. degree from National Chiao Tung University, Taiwan, R.O.C. in 1986, the M.S. degree from National Taiwan University in 1988, and the Ms. Sci. and Ph. D. degrees from the Georgia Institute of Technology , Atlanta, in 1995, and 1996, respectively, all in electrical engineering.

From 1990 to 1992, he was with the Telecommunications Laboratories of Chunghwa Telecom Co. In 1995, he was affiliated with Bell Northern Research of Northern Telecom, Inc., Richardson, TX. From 1996 to 2000, he was with AT\&T Laboratories, where he was a Senior Technical Staff Member in the Wireless Communications Research Department.  Since August 2000, he has joined the Department of Electrical and Computer Engineering of National Chiao Tung University in Taiwan and is the current Chairman of the same department. His current research interests are in the areas of radio resource management and cross-layer optimization techniques for wireless systems, heterogeneous wireless network design, and cloud computing for mobile applications.

Dr. Wang won the Distinguished Research Award of National Science Council, Taiwan in 2012, and was elected to the IEEE Fellow grade in 2011 for his contributions to cellular architectures and radio resource management in wireless networks. He was a co-recipient(with Gordon L. Stuber and Chin-Tau Lea) of the 1997 IEEE Jack Neubauer Best Paper Award for his paper ``Architecture Design, Frequency Planning, and Performance Analysis for a Microcell/Macrocell Overlaying System," IEEE Transactions on Vehicular Technology, vol. 46, no. 4, pp. 836-848, 1997. He has published over 200 journal and international conference papers. He served as an Associate Editor for the IEEE Trans. on Wireless Communications from 2001 to 2005, the Guest Editor of Special Issue on "Mobile Computing and Networking" for IEEE Journal on Selected Areas in Communications in 2005, "Radio Resource Management and Protocol Engineering in Future Broadband Networks" for IEEE Wireless Communications Magazine in 2006, and "Networking Challenges in Cloud Computing Systems and Applications," for IEEE Journal on Selected Areas in Communications in 2013, respectively. He is holding 10 US patents.
\end{biography}

\end{document}